\documentclass[letterpaper, 10 pt, conference]{ieeeconf}
\IEEEoverridecommandlockouts
\usepackage{graphicx}
\usepackage[dvipsnames]{xcolor}
\usepackage{amsfonts,amsmath,amssymb}
\usepackage[utf8]{inputenc}
\usepackage{epstopdf}
\usepackage{soul}
\usepackage{url}
\usepackage{todonotes}
\usepackage{float}
\usepackage[font=small]{caption, subcaption}
\usepackage[draft,inline,nomargin,index]{fixme}

\definecolor{forest}{HTML}{009933}
\definecolor{bluesea}{HTML}{6666ff}
\definecolor{rederror}{HTML}{ff3333}
\fxsetup{theme=color,mode=multiuser}
\FXRegisterAuthor{lg}{lgenv}{\color{forest}LG} 
\FXRegisterAuthor{mb}{mbenv}{\color{bluesea}MB} 
\FXRegisterAuthor{lm}{lmenv}{\color{rederror}LM} 

\newcommand{\R}{\mathbb{R}}
\newcommand{\N}{\mathbb{N}}
\newcommand{\cH}{\mathcal{H}}
\newcommand{\cS}{\mathcal{S}}
\newcommand{\cJ}{\mathcal{J}}
\newcommand{\cK}{\mathcal{K}}
\newcommand{\C}{\mathcal{C}}
\newcommand{\D}{\mathcal{D}}
\newcommand{\B}{\mathcal{B}}
\newcommand{\KL}{\mathcal{KL}}
\newcommand{\cL}{\mathcal{L}}
\newcommand{\kF}{\kappa}
\newcommand{\kM}{\boldsymbol{\mathcal{K}}}
\newcommand{\kV}{\boldsymbol{\kappa}}
\newcommand{\bg}{\boldsymbol{\gamma}}
\newcommand{\cvar}{z}
\newcommand{\sty}{\boldsymbol{\dot{\eta}_d}}
\newcommand{\stx}{\boldsymbol{\eta}}
\newcommand{\bcvar}{\boldsymbol{\cvar}}
\newcommand{\T}{^{\top}}

\newcommand\figref{Figure~\ref}
\newcommand\secref{Section~\ref}
\newcommand\tabref{Table~\ref}
\newcommand\lref{Lemma~\ref}
\newcommand\assmref{Assumption~\ref}
\newcommand\col{\text{col}}

\newtheorem{theorem}{Theorem}[section]

\newtheorem{lemma}[theorem]{Lemma}

\newtheorem{assumption}[theorem]{Assumption}

\title{\LARGE Adaptive Nonlinear Regulation via Gaussian Process}%
\author{Lorenzo Gentilini$^1$, Michelangelo Bin$^2$, and Lorenzo Marconi$^1$%
\thanks{$^1$L. Gentilini and L. Marconi are with the Center for Research on Complex Automated Systems (CASY), Department of Electrical, 
Electronic and Information Engineering (DEI), University of Bologna, Bologna, Italy (e-mails:
\texttt{\{lorenzo.gentilini6, lorenzo.marconi\}@unibo.it}). $^2$Michelangelo Bin is with the Department of Electrical and Electronic Engineering, Imperial College London,
London, UK (e-mail: \texttt{m.bin@imperial.ac.uk}).}}%

\begin{document}
\maketitle
\thispagestyle{empty}
\pagestyle{empty}

\begin{abstract}
   The paper deals with the problem of output regulation of nonlinear systems by presenting a learning-based adaptive
   internal model-based design strategy. We borrow from the adaptive internal model design technique recently proposed
   in~\cite{bin2019class} and extend it by means of a Gaussian process regressor. The learning-based adaptation is
   performed by following an ``event-triggered'' logic so that hybrid tools are used to analyse the resulting closed-loop
   system. Unlike the approach proposed in~\cite{bin2019class} where the friend is supposed to belong to a specific
   finite-dimensional model set, here we only require smoothness of the ideal steady-state control action.
   The paper also presents numerical simulations showing how the proposed method outperforms previous approaches. 
\end{abstract}

\section{Introduction}%
\label{sec:INTRODUCTION}
We consider a class of nonlinear systems of the form
\begin{equation}%
   \label{eq:CLASS-SYSTEM}
   \begin{matrix}
      \dot{x} = f(x,w,u), & y = h(x,w), & e = h_e(x,w),
   \end{matrix}
\end{equation}
with state $x \in \R^{n_x}$, control input $u \in \R^{n_u}$, measured outputs $y \in \R^{n_y}$, regulation error $e \in \R^{n_e}$,
and with $w \in \R^{n_w}$ an exogenous signal.
As customary in the literature of output regulation we refer to the subsystem $w$ as the \textit{exosystem}
\begin{equation}%
   \label{eq:GENERAL-EXOSYSTEM}
   \dot{w} = s(w).
\end{equation}
For the class of systems~\eqref{eq:CLASS-SYSTEM}, in this paper we consider the problem of design an
output-feedback regulator of the form
\begin{equation*}
   \begin{matrix}
      \dot{x}_c = f_c(x_c, y, e), & u = k_c(x_c, y, e)
   \end{matrix}
\end{equation*}
that ensures boundedness of the closed-loop trajectories and asymptotically removes the effect of $w$
on the regulated output $e$, thus \textit{ideally} obtaining $e(t) \to 0$ as $t \to \infty$.
More precisely, the seeked regulator ensures
\begin{equation*}
   \lim \sup_{t \rightarrow \infty} \left\|e(t)\right\| \le \varepsilon
\end{equation*}
with $\varepsilon \ge 0$ possibly a small number measuring the regulator's asymptotic performance.
Depending on the control objective the regulation problem undergoes the following taxonomy.
\textit{Asymptotic regulation} denotes the case in which $\varepsilon = 0$.
\textit{Approximate regulation} denotes the case in which the control objective is relaxed with a fixed $\varepsilon > 0$.
Finally, \textit{practical regulation} refers to the case in which $\varepsilon$ can be reduced arbitrary by tuning the regulator parameters.
When one of the above control objectives is achieved in spite of uncertainties in the plant's model,
we call it \textit{robust regulation}, moreover if some learning mechanism is introduced to compensate for
uncertainties in the exosystem, the problem is typically referred to as \textit{adaptive regulation}.
This work frames specifically the problem of ``adaptive approximate'' regulation, where the adaptation side
is approached in a system identification fashion with the Gaussian process regression used to infer the internal model dynamics
directly out of the collected data.

\subsubsection*{Related Works}
Asymptotic output regulation is a rich research area with a well-established theoretical foundation.
Although the number of solutions presented in the literature is quite variegated, both for linear and nonlinear systems,
most of the work can be traced-back to~\cite{francis1976internal} and~\cite{davison1976robust} where
Francis, Wonham, and Davison firstly formalize and solve the asymptotic regulation problem in the context of linear systems.
Asymptotic output regulation for Single-Input-Single-Output (SISO) nonlinear systems has been
under investigation since the early 90s, first in a local context~\cite{byrnes1997structurally},~\cite{isidori1990output},
and later in a purely nonlinear framework~\cite{byrnes2003limit},~\cite{marconi2007output},
based on the ``non-equilibrium'' theory~\cite{byrnes2003asymptotic}.
Recently, asymptotic regulators have been also extended to some classes of multivariable nonlinear
systems~\cite{wang2016nonlinear},~\cite{wang2017robust}.
The major limitation of such regulators dwells in the complexity. Indeed, the sufficient conditions under
which asymptotic regulation is ensured are typically expressed by equations whose analytic solution becomes a hard task even for ``simple'' problems.
Moreover, even if a regulator can be constructed, asymptotic regulation remains a fragile property that is lost at front of
the slightest plant's or exosystem's perturbation~\cite{bin2018robustness}.
The aforementioned problems motivates the researcher to move toward more robust solutions, introducing the concept of adaptive and approximate regulation.
Among the approaches to approximate regulation it is worth mentioning~\cite{marconi2008uniform} and~\cite{astolfi2015approximate}, whereas practical regulators
can be found in~\cite{isidori2012robust} and~\cite{freidovich2008performance}. Adaptive designs of regulators can be found
in~\cite{priscoli2006new} and~\cite{pyrkin2017output}, where linearly parametrized
internal models are constructed in the context of adaptive control, in~\cite{bin2019adaptive} where discrete-time adaptation algorithms
are used in the context of multivariable linear systems, and in~\cite{forte2016robust},~\cite{bin2019class},~\cite{bin2020approximate}, where adaptation of a nonlinear
internal model is approached as a system identification problem.

Learning dynamics models is also an active research topic.
In particular, Gaussian Processes (GPs) are increasingly used to estimate unknown dynamics~\cite{kocijan2016modelling},~\cite{buisson2020actively}.
Unlike other nonparametric models, GPs represent an attractive tool in learning dynamics due to their flexibility in modeling nonlinearities and
the possibility to incorporate prior knowledge~\cite{rasmussen2003gaussian}.
Moreover, since GPs allow for analytical formulations, theoretical guarantees on the a posteriori can be drawn directly from the collected
data~\cite{umlauft2020learning},~\cite{lederer2019uniform}.
Recently, GP models spread inside the field of nonlinear optimal control~\cite{sforni2021learning}, with several applications to the particular case of
Model Predictive Control (MPC)~\cite{torrente2021data},~\cite{kabzan2019learning}, and inside the field of nonlinear observers~\cite{buisson2021joint}.

To the best of the authors knowledge, this is the first attempt to couple an internal model regulator with a GP regressor,
bridging the gap between adaptive regulation and learning dynamics.

\subsubsection*{Contributions}
In this paper we propose a novel learning-based adaptive regulation technique built on top of the recently published work~\cite{bin2019class},
where a ``class-type'' identification-based internal model is used to solve the problem of \textit{approximate regulation}.
In particular, we borrow the proposed post-processing internal model and extend it by means of a Gaussian process regressor.
While the plan and the controller evolve in continuous time, the proposed identifier is updated at specific discrete time instants
in an event-triggered fashion. Hence the closed-loop system is a \textit{hybrid system}.
The additional complexity introduced in the analysis is motivated by the ability of the regulator to generate a larger
``class of input signals'' needed to ensure zero regulation error (the so-called \textit{friend}~\cite{bin2018chicken}).
Unlike previous approaches where the friend is supposed to belong to a specific finite-dimensional model set~\cite{bin2019class}, here we only
require smoothness, namely we assume that the steady-state control action belongs to a predefined Reproducing Kernel Hilbert Space (RKHS).
As the reported numerical simulations show, the proposed approach presents comparable performance to~\cite{bin2019class} when the friend can be well characterized by
the chosen model set, while outperforms previous approaches otherwise.
The proposed setting falls outside the framework of~\cite{bin2019class}, which only deals with continous-time identifiers, and
results to be closer to the approach of~\cite{bin2020approximate} where, however, a different regulator structure is employed.
Nevertheless, by means of the same arguments used in~\cite{bin2020approximate}, it is possible to prove that hybrid identifiers can be used in the
regulator of~\cite{bin2019class} as well, provided that the identifier satisfies some strong stability properties formalized as the 
\textit{identifier requirement}~\cite[Definition 1]{bin2020approximate}. 
In this work, we construct a Gaussian process-based identifier satisfying such an identifier requirement, therefore fitting the
framework of~\cite{bin2019class} and~\cite{bin2020approximate}.
This allows us to use the designed GP-based identifier to solve adaptive regulation problems for a class of nonlinear systems.

The paper unfolds as follows.
In~\secref{sec:PROBLEM-AND-PRELIMINARIES} we describe the problem set-up, along with the standing assumptions, and some preliminaries.
In~\secref{sec:LEARNING-DRIVEN-REGULATION} we present the proposed regulator and state the main result of the paper.
Finally, in~\secref{sec:NUMERICAL-SIMULATIONS} a numerical example is presented.

\section{Problem set-up and Preliminaries}%
\label{sec:PROBLEM-AND-PRELIMINARIES}
In this section we first present in detail the approximate regulation problem that this work focuses on, along with
the basic standing assumptions. Then, a post-processing internal model design technique and the basic concepts
behind the notion of Gaussian process regression are reviewed.
\subsection{Nonlinear Approximate Output Regulation}
We focus on a subclass of the general regulation problem
presented in~\secref{sec:INTRODUCTION}, by considering systems of the form
\begin{equation}%
   \label{eq:PARTICULAR-SYSTEM}
   \begin{split}
      \dot{x}_0 &= f_0(x,w) + b(x,w)u\\
      \dot{\chi} &= F \chi + H \zeta\\
      \dot{\zeta} &= q(x,w) + \Omega(x,w)u\\
      e &= C \chi, \hspace{0.5cm} y = \text{col}(\chi, \zeta),
   \end{split}
\end{equation}
in which $x_0 \in \R^{n_0}$, $y \in \R^{n_y}$, $e \in \R^{n_e}$, $\chi \in \R^{n_e}$, and, $u \in \R^{n_u}$ with $n_u \ge n_e$.
Moreover, $\chi = \col(\chi^1, \dots, \chi^{n_e})$, with $\chi^i \in \R^{n_{\chi}^i}$, $i = 1, \dots, n_e$, and
$\sum_{k=1}^{n_e} n_{\chi}^k = n_{\chi}$. The matrices $F \in \R^{n_{\chi} \times n_{\chi}}$, $H \in \R^{n_{\chi} \times n_{e}}$,
and $C \in \R^{n_{e} \times n_{\chi}}$ are defined as a block-diagonal matrices with entries
\begin{equation*}
   \begin{matrix}
      F_i =
      \begin{pmatrix}
         0_{(n_{\chi}^i-1) \times 1} & I_{n_{\chi}^i - 1} \\
         0 & 0_{1 \times (n_{\chi}^i-1)}
      \end{pmatrix}, &
      H_i =
      \begin{pmatrix}
         0_{(n_{\chi}^i-1) \times 1} \\ 1
      \end{pmatrix},
   \end{matrix}
\end{equation*}
\begin{equation*}
   C_i =
      \begin{pmatrix}
         1 & 0_{1 \times (n_{\chi}^i-1)}
      \end{pmatrix}.
\end{equation*}
Equation~\eqref{eq:PARTICULAR-SYSTEM} frames the problem of output regulation on a particular class of systems that embraces
a large number of use-cases addressed in literature. In particular, note that all systems presenting \textit{(a)} a well-defined vector
relative degree and admitting a canonical normal form, or that are \textit{(b)} strongly invertible and feedback linearisable~\cite{wang2014global},
with respect to the pair $\left(u,e\right)$, fit inside the proposed framework.
The results presented in this work are based on the following standing assumptions (see~\cite[Assumption~A1,~A2]{bin2019class}),
where we denote the state of~\eqref{eq:PARTICULAR-SYSTEM} by $x = \col(x_0, \chi, \zeta)$.
\begin{assumption}%
   \label{assm:REGULATOR-EQUATIONS}
   There exist $\beta_0 \in \KL$, $\alpha_0 > 0$ and, for each solution $w$ of~\eqref{eq:GENERAL-EXOSYSTEM}, there exist $x_0^* : \R_{\ge 0} \mapsto \R^{n_0}$
   and $u^* : \R_{\ge 0} \mapsto \R^{n_u}$ fulfilling
   \begin{equation}%
      \label{eq:REGULATOR-EQUATIONS}
      \begin{split}
         \dot{x}_0^* & = f_0(w, x^*) + b(w, x^*)u^*, \\
         0 & = q(w, x^*) + \Omega(w, x^*)u^*,
      \end{split}
   \end{equation}
   where $x^* = \left(x_0^*,0,0\right)$, and for all $t > 0$ the following holds
   \begin{equation*}
      \left| x_0(t) - x_0^*(t) \right| \le \beta_0\left(\left| x_0(0) - x_0^*(0) \right|, t\right) + \alpha_0 \left| \left(\chi, \zeta\right) \right|_{\left[0,t\right)}.
   \end{equation*}
\end{assumption}
\begin{assumption}%
   \label{assm:UNIFORM-DETECTABILITY}
   There exists a full-rank matrix $\cL \in \R^{n_u \times n_e}$ such that the matrix $\Omega(w, x)\cL$ is bounded, and
   \begin{equation*}
      \cL\T \Omega(w,x)\T + \Omega(w,x) \cL \ge I_{n_e}
   \end{equation*}
   holds for all $(w,x) \in \R^{n_w} \times \R^{n_x}$, and the map $(w,x) \mapsto (\Omega(w,x)\cL)^{-1} q(w,x)$ is Lipschitz.
\end{assumption}
In this framework,~\cite{bin2019class} proposes a post-processing internal model of the form
\begin{equation}%
   \label{eq:INTERNAL-MODEL}
   \begin{matrix}
      \dot{\eta} = \Phi(\eta) + Ge, & \eta \in \R^{dn_e},
   \end{matrix}
\end{equation}
with $d \in \N$, $\eta = \left(\eta_1, \dots, \eta_d\right)\T$, $\eta_i \in \R^{n_e}$, and
\begin{equation*}
   \begin{matrix}
      \Phi(\eta) = 
      \begin{pmatrix}
         \eta_2 \\ \vdots \\ \eta_d \\ \psi(\eta)
      \end{pmatrix}, &
      G =
      \begin{pmatrix}
         g h_1 I_{n_e} \\ g^2 h_2 I_{n_e} \\ \vdots \\ g^d h_d I_{n_e}
      \end{pmatrix}.
   \end{matrix}
\end{equation*}
In the aforementioned definition, the coefficients $h_i$, with $i = 1, \dots, d$, are fixed so that the polynomial
$s^d + h_1s^{d-1}+\cdots+h_{d-1}s+h_d$ is Hurwitz, $g > 0$ is a parameter to be designed, and
$\psi: \R^{dn_e} \mapsto \R^{n_e}$ is a function to be fixed.
Moreover, the static stabiliser control action is chosen as
\begin{equation}%
   \label{eq:CONTROL-ACTION}
   u = \cL(K_{\chi}\chi + K_{\zeta}\zeta + K_{\eta}\eta_1 + K_w \nu(x^*, w)),
\end{equation}
where the matrices $K_{\chi}$, $K_{\zeta}$, and $K_{\eta}$ take the form
\begin{equation*}
   \begin{matrix}
      K_{\chi}(l,\delta) = lK(\delta), & K_{\zeta}(l) = -lI_{n_e}, & K_{\eta}(l,\delta) = lK(\delta)C\T,
   \end{matrix}
\end{equation*}
with $K(\delta) = \text{blkdiag}\left(K^1(\delta), \dots, K^{n_e}(\delta)\right)$, where
\begin{equation*}
   K^i(\delta) = - \left( \begin{matrix} c_1^i \delta^{n_{\chi}^i} & c_2^i \delta^{n_{\chi}^i - 1} & c_{n_{\chi}^i}^i \delta \end{matrix} \right),
\end{equation*}
for $i = 1, \dots, n_e$, in which the coefficients $c_j^i$ are chosen so that the polynomials
$s^{n_{\chi}^i} + c^i_{n_{\chi}^i}s^{n_{\chi}^i - 1} + \cdots + c_2^i s + c_1^i$, $i = 1, \dots, n_e$, are Hurwitz, and $l, \delta > 0$ are design parameters to be fixed.
Note that the matrix $K_{w}$ and the function $\nu(\cdot, \cdot)$ are left as a degree of freedom. Indeed these quantities can be used to represent possible feedforward
contributions added by the designer employing knowledge about $w$ and $x^*$ (Equation~\eqref{eq:REGULATOR-EQUATIONS}).
For further details, the reader is referred to~\cite{bin2019class}.

A key step during the regulator design is represented by the selection of the pair $\left( d, \psi \right)$, which
should be chosen in order to achieve small, possibly zero, asymptotic regulation error, in spite of uncertainties involving the
regulation equations~\eqref{eq:REGULATOR-EQUATIONS} and the system dynamics~(\ref{eq:GENERAL-EXOSYSTEM},~\ref{eq:PARTICULAR-SYSTEM}).
Such an ambitious goal can be reached only by relying on some preliminary knowledge about the \textit{class of signals} to which
$\dot{\eta}_d$ and $\eta$ are expected to belong. In this context the steady-state signals $(x^*, u^*)$ are the anchor point from which that knowledge can be drawn.
In particular, letting $\eta_1^* = \bg_{(l,\delta)}(w,x^*)$ from Equation~\eqref{eq:CONTROL-ACTION} with $u^*$ defined as Equation~\eqref{eq:REGULATOR-EQUATIONS},
the ideal internal model trajectory is
\begin{equation*}
   \begin{split}
      \eta_i^* & = L_{s(w)}^{i-1} \bg_{(l,\delta)}(w,x^*) + L_{f_0(w,x^*)+b(w,x^*)u^*}^{i-1}\bg_{(l,\delta)}(w,x^*), \\
      \dot{\eta}_d^* & = \eta_{d+1}^*,
   \end{split}
\end{equation*}
with $i = 2, \dots, d+1$ and $L^j_{h(\cdot)}f(\cdot)$ denoting the $j$th Lie derivative of $f$ along the trajectories of $h$.
The pair $\left( d^*, \psi^* \right)$ should be ideally chosen to guarantee $\dot{\eta}_{d^*}^* = \psi^*(\eta^*)$ which makes
$(x^*, \eta^*)$ a trajectory of the closed-loop system with associated zero regulation error.
Nevertheless, the a priori design of such couple is not realistic since the solution of Equation~\eqref{eq:REGULATOR-EQUATIONS} is usually
highly uncertain and the quantities $\eta^*$ depend on the stabiliser structure $(l, \delta)$ that cannot be apriori fixed,
raising the so-called \textit{chicken-egg dilemma}~\cite{bin2018chicken}.
Current state-of-the-art strategies try to handle the problem of designing a suitable pair $\left( d, \psi \right)$
by casting it as an \textit{identification problem}, where the map $\psi$ is interpreted as a \textit{prediction model}
parameterized via a set of auxiliary parameters $\theta \in \R^{n_{\theta}}$.
In doing that, one implicitly assumes that the \textit{friend} belongs to a known fixed \textit{model set} of the form
\begin{equation*}
   \mathcal{M} = \{ \psi_{\theta} : \theta \in \R^{n_{\theta}} \}.
\end{equation*}
Therefore, these approaches present the disadvantage of limiting the \textit{class of friends} which we can deal with, leading to a
degraded performance in all those cases in which the steady-state signals $(x^*, u^*)$ are highly uncertain
and the chosen class $\mathcal{M}$ is inadequate to represent the ideal function $\psi^*$.
Unlike previous works in this field, we drop the assumption of $\psi$ belonging to a given \textit{model set} on
behalf of a more general and less conservative hypothesis. In particular, we let such a function be of whatever shape, with
the only constraint to be \textit{sufficiently smooth}.
In view of the latter, we recall~\cite[Assumption~A3]{bin2019class} under which the asymptotic stability results can be drawn.
\begin{assumption}%
   \label{assm:INTERNAL-MODEL}
   The map $\psi(\eta)$ is Lipschitz and differentiable with a locally Lipschitz derivative, and the Lipschitz constants do not depend on $\delta$ and $l$.
   Moreover, there exists a compact set $H^* \subset \R^{n_e} \times \R^{dn_e}$, independent on $\delta$ and $l$, such that every solution of~\eqref{eq:INTERNAL-MODEL}
   satisfies $\left(\eta_d^*(t), \eta^*(t) \right) \in H^*$ for all $t \in \R_{\ge 0}$.
\end{assumption}

\subsection{Gaussian Process Regression}%
\label{sec:GAUSSIAN-INFERENCE}
The key idea behind the proposed approach dwells in modeling the unknown function $\psi$ as the realization of a Gaussian process.
GPs are function estimators widely used because of the flexibility they offer in modeling nonlinear maps directly out from the collected data~\cite{rasmussen2003gaussian}.
A GP model is fully described by a mean function $m: \R^{dn_e} \mapsto \R^{n_e}$ and a covariance function (\textit{kernel})
$\kF: \R^{dn_e} \times \R^{dn_e} \mapsto \R$. Whereas there are many possible choices of mean and covariance functions, in this work we keep 
the formulation of $\kF$ general, with the only constraint expressed by~\assmref{assm:K-CONTINUOUS-BOUNDED} below. Yet we force, without loss of generality, $m(\eta) = 0_{n_e}$
for any $\eta$. Thus we assume that
\begin{equation*}
   \psi(\eta) \sim \mathcal{GP}(0, \kF(\cdot, \cdot)).
\end{equation*}
Supposing to have access to a data-set $(\stx, \sty) = \{ (\eta(t_1), \dot{\eta}_d(t_2)), \dots, (\eta(t_{n_{\text{ds}}}), \dot{\eta}_d(t_{n_{\text{ds}}}))\}$
with each pair $(\eta(t_h), \dot{\eta}_d(t_h)) \in \R^{dn_e} \times \R^{n_e}$ obtained as $\dot{\eta_d}(t_h) = \psi(\eta(t_h)) + \varepsilon(t_h)$ with
$\varepsilon(t_h) \sim \mathcal{N}(0, \sigma_n^2 I_{n_e})$ be a white Gaussian noise, then the posterior distribution of $\psi$ given the data-set is still
a Gaussian process with mean $\mu$ and variance $\sigma^2$ given by (\cite{rasmussen2003gaussian})
\begin{equation}%
   \label{eq:GP-POSTERIOR}
   \begin{split}
      \mu(\eta) & = \kV(\eta)\T\left(\kM + \sigma_n^2 I_{n_{\text{ds}}}\right)^{-1}\sty, \\
      \sigma^2(\eta) & = \kF(\eta, \eta) - \kV(\eta)\T\left(\kM + \sigma_n^2 I_{n_{\text{ds}}}\right)^{-1}\kV(\eta),
   \end{split}
\end{equation}
where $\kM \in \R^{n_{\text{ds}} \times n_{\text{ds}}}$ is the \textit{Gram matrix} whose $(k,h)$th entry is $\kM_{k,h} = \kF(\stx_k, \stx_h)$,
with $\stx_k$ the $k$th entry of $\stx$, and $\kV(\eta) \in \R^{n_{\text{ds}}}$ is the kernel vector whose $k$th component is $\kV_k(\eta) = \kF(\eta, \stx_k)$.
The problem of inferring an unknown function $\psi$ from a finite set of data can be seen as a special case of \textit{ridge regression} where the
prior assumptions (mean and covariance) are encoded in terms of $smoothness$ of $\mu$.
In particular, let $\cH$ be a RKHS associated with the kernel function $\kF$, then the function $\psi$ can be infered by minimizing the functional
\begin{equation*}
   \cJ = \frac{\lambda_s}{2} \left\| \mu \right\|^2_{\cH} + Q(\sty, \mu(\stx)),
\end{equation*}
where the first term plays the role of \textit{regularizer} and represents the smoothness assumptions on $\mu$
as encoded by a suitable RKHS, while the second one represents the data-fit term assessing the quality of the prediction
$\mu(\stx)$ with respect to the observed data $\sty$~\cite{rasmussen2003gaussian}.
According to the \textit{representer theorem}~\cite{o1986automatic}, each minimizer $\mu \in \cH$ of $\cJ$ takes the form  
$\mu(\eta) = \kV(\eta)\alpha$. In the particular case in which $Q(\sty, \mu(\stx))$ corresponds to a negative log-likelihood
of a Gaussian model with variance $\sigma_n^2$, namely
\begin{equation*}
   Q(\sty, \mu(\stx)) = \frac{1}{2 \sigma_n^2} \left\| \sty - \mu(\stx) \right\|^2_2,
\end{equation*}
the value of $\alpha$ recovers the expression in Equation~\eqref{eq:GP-POSTERIOR} as
\begin{equation*}
   \alpha = (\kM + \sigma_n^2I_{n_{\cvar}})^{-1}\sty.
\end{equation*}
From now on we suppose that the following standing assumptions hold (see~\cite[Assumption~2,~Assumption~3]{buisson2021joint})
\begin{assumption}%
   \label{assm:MU-CONTINUOUS-BOUNDED}
   $\mu$ is Lipschitz continuous with Lipschitz constant $L_{\eta}$, and its norm is bounded by $\mu_{\text{max}}$.
\end{assumption}
\begin{assumption}%
   \label{assm:K-CONTINUOUS-BOUNDED}
   The kernel function $\kF(\cdot, \cdot)$ is Lipschitz continuous with constant $L_{\kF}$, with a locally Lipschitz derivative of constant $L_{d\kF}$,
   and its norm is bounded by $\kF_{\text{max}}$.
\end{assumption}

Although any kernel fulfilling~\assmref{assm:K-CONTINUOUS-BOUNDED} can be a valid candidate, in the following,
we exploit the commonly adopted \textit{squared exponential kernel} as prior covariance function, which can be expressed as
\begin{equation}%
   \label{eq:EXPONENTIAL-KERNEL}
   \kF(\eta, \eta') = \sigma^2_p \exp\left( -\left(\eta - \eta'\right)\T \Lambda^{-1}  \left(\eta - \eta'\right) \right)
\end{equation}
for all $\eta, \eta' \in \R^{dn_e}$, where $\Lambda = \text{diag}(2\lambda_{\eta_1}^2, \dots, 2\lambda_{\eta_{dn_e}}^2)$, $\lambda_{\eta_{i}} \in \R_{>0}$ is
known as \textit{characteristic length scale} relative to the $i$th signal, and $\sigma^2_p$ is usually called
\textit{amplitude}~\cite{rasmussen2003gaussian}.
We conclude this section by stating a constructive assumption, on which the main contribution of this work is built.
\begin{lemma}%
   \label{lem:RKHS-LIPSCHITZ}
   Let $\cH$ be a RKHS induced by a positive-definite kernel $\kF$ and let $f \in \cH$,
   then $f$ is Lipschitz continuous with constant $\left\| f \right\|_{\cH}$.
\end{lemma}
The proof of the above lemma follows directly from the application of the reproducing property
and the Cauchy-Schwartz inequality, and thus omitted.
\begin{assumption}%
   \label{assm:RKHS-BELONGING}
   The map $\psi$ belongs to the RKHS associated to the kernel function $\kF(\cdot, \cdot)$ in Equation~\eqref{eq:EXPONENTIAL-KERNEL}.
\end{assumption}
It is worth noting that the first statement of~\assmref{assm:INTERNAL-MODEL} is implied by~\assmref{assm:RKHS-BELONGING} by means of
\lref{lem:RKHS-LIPSCHITZ}.

\section{Gaussian Process-based Adaptive Regulation}%
\label{sec:LEARNING-DRIVEN-REGULATION}
The proposed regulator reads as follows
\begin{equation}%
   \label{eq:OVERALL-REGULATOR}
   \begin{split}
      &
      \begin{cases}
         \dot{\tau} & = 1 \\
         \dot{\eta} & = \Phi(\eta, \mu(\eta, \varsigma, \alpha, \tau)) + Ge \\
         \dot{\xi}_1 & = \xi_2 - m_1 \rho(\xi_1 - \eta_d) \\
         \dot{\xi}_2 & = \dot{\mu}(\eta, \varsigma, \alpha, \tau) - m_2 \rho^2 (\xi_1 - \eta_d) \\
         \dot{\varsigma} & = 0 \\
      \end{cases} \\
      & (\eta, \xi_1, \xi_2, \varsigma, \tau) \in \C, \\
      &
      \begin{cases}
         \tau^{+} & = 0 \\
         \eta^{+} & = \eta \\
         \xi_1^{+} & = \xi_1 \\
         \xi_2^{+} & = \xi_2 \\
         \varsigma^{+} & = \left(S \otimes I_{p}\right) \varsigma + \left(B \otimes I_{p}\right) \begin{bmatrix} \eta & \xi_2 & \tau \end{bmatrix}\T \\
      \end{cases} \\
      & (\eta, \xi_1, \xi_2, \varsigma, \tau) \in \D, \\
   \end{split}
\end{equation}
in which $\alpha = \gamma(\varsigma) \in \R^{n_{\text{ds}} \times n_{\text{ds}}}$, $p = dn_e + n_e + 1$,
$(\xi_1, \xi_2) \in \R^{n_e} \times \R^{n_e}$, and $\varsigma \in \R^{n_{\varsigma}}$ with $n_{\varsigma} = pn_{\text{ds}}$.
The matrices $S \in \R^{n_{\text{ds}} \times n_{\text{ds}}}$ and $B \in \R^{n_{\text{ds}}}$ have the ``shift'' form
\begin{equation*}
   \begin{matrix}
      S = 
      \begin{pmatrix}
         0_{(n_{\text{ds}}-1) \times 1} & I_{n_{\text{ds}}-1} \\
         0 & 0_{1 \times (n_{\text{ds}}-1)}
      \end{pmatrix}, &
      B = 
      \begin{pmatrix}
         0_{(n_{\text{ds}}-1) \times 1} \\ 1
      \end{pmatrix},
   \end{matrix}
\end{equation*}
while $\Phi$ and $G$ have the same structure described by Equation~\eqref{eq:INTERNAL-MODEL},
$\C = \{ (\eta, \xi_1, \xi_2, \varsigma, \tau) \in \R^{dn_e} \times \R^{n_e} \times \R^{n_e} \times \cS \times \R_{\ge 0} : \sigma^2(\eta, \varsigma, \alpha, \tau) \le \sigma^2_{\text{thr}} \}$
represents the flow set, and
$\D = \{ (\eta, \xi_1, \xi_2, \varsigma, \tau) \in \R^{dn_e} \times \R^{n_e} \times \R^{n_e} \times \cS \times \R_{\ge 0} : \sigma^2(\eta, \varsigma, \alpha, \tau) \ge \sigma^2_{\text{thr}} \}$
is the jump set, with $\sigma^2_{\text{thr}}$ arbitrary. The functions $\mu(\eta, \varsigma, \alpha, \tau)$ and $\sigma^2(\eta, \varsigma, \alpha, \tau)$
represent the a posteriori GP estimated mean and variance after the collection of $n_{\text{ds}}$ samples, and
$\cS \subseteq \R^{n_{\text{ds}} n_{\varsigma}}$.
The proposed regulator is composed of \textit{(a)} a purely continuous-time subsystem $\left( \eta, \xi_1, \xi_2 \right)$ whose dynamics depends
on $\varsigma$ and $\alpha$ that are constant during flow, \textit{(b)} a purely discrete-time subsystem $\varsigma$ updated ad each jump time,
and \textit{(c)} a hybrid clock $\tau$. Note that, due to the definition of the sets $\C$ and $\D$ the jumping law is not directly
related to the clock variable $\tau$, thus at the moment it is not clear if~\eqref{eq:OVERALL-REGULATOR} suffers from
\textit{zeno} or \textit{chattering} issues.
The subsystem $\eta$ plays the role of internal model unit, and it is taken of the same form as~\eqref{eq:INTERNAL-MODEL}.
The subsystem $(\xi_1, \xi_2)$ plays the role of \textit{observer} of the quantity $\dot{\eta}_d$ required to build the 
data-set and not directly available. The dynamic equation of $(\dot{\xi}_1, \dot{\xi}_2)$ follows the canonical high-gain construction
with the coefficients $m_1, m_2 > 0$ arbitrary and $\rho > 0$ left as a design parameter.
The quantities $(\xi_2, \eta)$ act as proxy for the ideal $(\dot{\eta}_d^*, \eta^*)$ required to build an approximation of $\psi^*$;
$(\xi_2, \eta)$ are thus feeded to the discrete-time GP regressor represented by the subsystem $\varsigma$ and by the properly defined
functions $(\mu, \sigma^2, \gamma)$. 
In particular, denoting $\cvar = \col{(\eta, \tau)}$, the latter functions read as
\begin{equation*}
   \begin{split}
      \mu(\cvar) & = \kV(\cvar)\T \gamma(\varsigma) \boldsymbol{\xi_2}, \\
      \sigma^2(\cvar) & = \kF(\cvar, \cvar) - \kV(\cvar)\T \gamma(\varsigma) \kV(\cvar), \\
      \gamma(\varsigma) & = \left( \kM + \sigma_n^2 I_{n_{\text{ds}}} \right)^{-1},
   \end{split}
\end{equation*}
where $\kV$ and $\kM$ are defined as in~\secref{sec:GAUSSIAN-INFERENCE}, with the only differency that the kernel function $\kF$
has been enhanced by adding a dependency from $\tau$. In this context, Equation~\eqref{eq:EXPONENTIAL-KERNEL} takes the form
\begin{equation*}
   \begin{split}
      \kF(\cvar_i, \cvar_j) & = \sigma^2_p \exp\left( -\left( \eta_i - \eta_j \right)\T \Lambda^{-1} \left( \eta_i - \eta_j \right) \right) \\
         & \exp \left( -\sum_{k = \min(i, j)}^{\max(i, j)} \left\| \frac{\tau_k}{2 \lambda_{\tau}} \right\| \right) \text{ if } \cvar_i, \cvar_j \in \bcvar,
   \end{split}
\end{equation*}
\begin{equation}%
   \label{eq:EXPONENTIAL-KERNEL-TIME-2}
   \begin{split}
      \kF(\cvar, \cvar_j) & = \sigma^2_p \exp\left( -\left( \eta - \eta_j \right)\T \Lambda^{-1} \left( \eta - \eta_j \right) \right) \\
         & \exp \left( - \frac{\tau}{2 \lambda_{\tau}^2} + \sum_{k = j}^{m} \left\| \frac{\tau_k}{2 \lambda_{\tau}} \right\| \right) \text{if } \cvar \notin \bcvar, \cvar_j \in \bcvar,
   \end{split}
\end{equation}
where $\lambda_{\tau}$ is a parameter to be tuned and the quantity $(\bcvar, \boldsymbol{\xi_2})$ represents the data-set constructed as discussed in~\secref{sec:GAUSSIAN-INFERENCE}
extended with the sampled clock $\tau$, and stored, as a state variable, inside the shift register $\varsigma$.
The addition of the clock as independent variable iside the GP regression is the key to obtain a good estimation of $\psi$ starting from
the noisy proxy $(\xi_2, \eta)$. This choice is motivated by the fact that the introduction of $\tau$ allows us to shape $\mu$
through the parameter $\lambda_{\tau}$ whose inverse value can be interpreted as a \textit{forgetting factor}, commonly used in identification.

The design parameters $\left( g, l, \delta, \rho \right)$ in~\eqref{eq:OVERALL-REGULATOR} can be chosen so that the closed-loop
system has an asymptotic regulation error bounded by a function of the best attainable prediction, namely
\begin{equation*}
   \lim_{t \to \infty} \sup \left\| e(t) \right\| = c_e \lim_{t \to \infty} \sup \left\| \dot{\eta}_d^*(t) - \mu(\eta^*(t)) \right\|,
\end{equation*}
with $c_e$ constant that depends on the chosen parameters.
Such a results directly follows from the adaptation of the arguments reported by~\cite{bin2019class} and~\cite{bin2020approximate}
in the specific case in which~\eqref{eq:OVERALL-REGULATOR} satisfies a set of properties known as \textit{identifier requirements}.
The following results are instrumental to verify that the proposed regulator fits inside the same framework of~\cite{bin2019class} and~\cite{bin2020approximate}.
\begin{lemma}%
   \label{lem:DWELL-TIME}
   Let Assumptions~\ref{assm:INTERNAL-MODEL},~\ref{assm:MU-CONTINUOUS-BOUNDED}, and~\ref{assm:K-CONTINUOUS-BOUNDED} hold.
   Moreover, suppose that the chosen $\sigma^2_{\text{thr}}$ satisfies
   \begin{equation}%
      \label{eq:SIGMA-CONDITION}
      \frac{\sigma_p^2 \sigma_n^2}{\sigma_p^2 + \sigma_n^2} < \sigma^2_{\text{thr}} < \sigma_p^2,
   \end{equation}
   then any solution $(\eta, \xi_1, \xi_2, \varsigma, \tau)$ of~\eqref{eq:OVERALL-REGULATOR} originating from the set
   $\chi_0 = \{(\eta, \xi_1, \xi_2, \varsigma, \tau) \in \R^{dn_e} \times \R^{n_e} \times \R^{n_e} \times \cS \times \R_{\ge 0} : \tau = 0\}$
   satisfies a dwell time condition in the sense of~\cite{hespanha1999stability}, and a reverse dwell time condition in
   the sense of~\cite{hespanha2005input}.
\end{lemma}
\begin{proof}
   Using the same arguments proposed by~\cite{hespanha2005input}, we shape the proof to show the existence of $\overline{\text{T}}$ and $\underline{\text{T}}$
   such that $\C \subseteq \{ (\eta, \xi_1, \xi_2, \varsigma, \tau) \in \R^{dn_e} \times \R^{n_e} \times \cS \times \R^{mp} \times \R_{\ge 0} : 0 \le \tau \le \overline{\text{T}} \}$
   and $\D \subseteq \{ (\eta, \xi_1, \xi_2, \varsigma, \tau) \in \R^{dn_e} \times \R^{n_e} \times \cS \times \R^{mp} \times \R_{\ge 0} : \underline{\text{T}} \le \tau \le \overline{\text{T}} \}$.
   First note that, in view of \assmref{assm:INTERNAL-MODEL}, from the definition of kernel in~\eqref{eq:EXPONENTIAL-KERNEL-TIME-2}, the value of $\sigma^2(\eta, \varsigma, \alpha, \tau)$
   reaches $\sigma_p^2$ as long as $t$ goes to infinty, i.e.
   \begin{equation}%
      \label{eq:SIGMA-LIMIT}
      \lim_{t \to \infty} \sigma^2(\eta(t), \varsigma(t), \alpha(t), \tau(t)) = \sigma_p^2.
   \end{equation}
   Recalling $\cvar = \col(\eta, \tau)$, then the flow and jump dynamics of $\sigma^2$ are described by
   \begin{equation}%
      \label{eq:SIGMA-FLOW}
      \begin{split}
         \dot{\left(\sigma^2\right)} & = \dot{\cvar}\T \left(\frac{d \kF(\cvar, \cvar)}{d \cvar} - 2 \frac{d \kV(\cvar)}{d \cvar}\T \gamma(\varsigma) \kV(\cvar) \right) \\
         & = - 2 \dot{\cvar}\T \frac{d \kV(\cvar)}{d \cvar}\T \gamma(\varsigma) \kV(\cvar),
      \end{split}
   \end{equation}
   when $\left(\eta, \xi_1, \xi_2, \varsigma, \tau\right) \in \C$, and
   \begin{equation}%
      \label{eq:SIGMA-JUMP}
      \left(\sigma^{2}\right)^+ = \kF(\cvar^+, \cvar^+) - \kV(\cvar^+)\gamma(\varsigma^+)\kV(\cvar^+),
   \end{equation}
   when $\left(\eta, \xi_1, \xi_2, \varsigma, \tau\right) \in \D$.
   The existence of $\overline{\text{T}}$ follows from~\eqref{eq:SIGMA-LIMIT} and the second inequality of~\eqref{eq:SIGMA-CONDITION},
   in particular by choosing $\sigma^2_{\text{thr}} < \sigma_p^2$, since $\sigma^2 \to_{\tau \to \infty} \sigma_p^2$, there exists a real value $\overline{\text{T}}$
   such that $\sigma^2 \ge \sigma^2_{\text{thr}}$ for any $\tau \ge \overline{\text{T}}$. This ensures persistency of jump intervals.
   Consider now Equations~\eqref{eq:SIGMA-FLOW} and~\eqref{eq:SIGMA-JUMP}, the flow dynamics $\dot{\left(\sigma^2\right)}$ results to be a continous function
   being it a product of Lipschitz functions (Assumptions~\ref{assm:INTERNAL-MODEL},~\ref{assm:MU-CONTINUOUS-BOUNDED},~\ref{assm:K-CONTINUOUS-BOUNDED}).
   Moreover, its norm can be upper bounded as
   \begin{equation*}
      \left\| \dot{\left(\sigma^2\right)} \right\| \le 2 \left\| \dot{\cvar} \right\| \left\| \frac{d \kV(\cvar)}{d \cvar} \right\| \left\| \gamma(\varsigma) \right\| \left\| \kV(\cvar) \right\|,
   \end{equation*}
   with $\dot{\cvar} = \left[ \eta_2, \dots, \eta_d, \psi(\eta), 1 \right]\T$.
   Since during flow the function $\kF$ takes the form of~\eqref{eq:EXPONENTIAL-KERNEL-TIME-2}, rewriting the latter as
   \begin{equation*}
      k(z, z_j) = \sigma_p^2 \exp\left(-(z-z_j)\T \bar{\Lambda}^{-1} (z-z_j)\right),
   \end{equation*}
   the quantity 
   \begin{equation*}
      \frac{d \kV(\cvar)}{d \cvar} =
      \begin{bmatrix}
         -\kF(\cvar, \cvar^1)\bar{\Lambda}^{-1}(\cvar - \cvar^1) \\
         -\kF(\cvar, \cvar^2)\bar{\Lambda}^{-1}(\cvar - \cvar^2) \\
         \vdots \\
         -\kF(\cvar, \cvar^m)\bar{\Lambda}^{-1}(\cvar - \cvar^m)
      \end{bmatrix},
   \end{equation*}
   where $\cvar^i$ denotes the $i$th training point stored inside $\varsigma$, and with $\bar{\Lambda}$ properly defined.
   \assmref{assm:INTERNAL-MODEL} ensures $\left\| \dot{\cvar} \right\| \le l_{\cvar}$, with $l_{\cvar}$ dependent on the chosen $H^*$,
   while~\assmref{assm:K-CONTINUOUS-BOUNDED} ensures boundedness of $\gamma(\varsigma)$ and $\kV(\cvar)$, namely
   \begin{equation*}
      \begin{matrix}
         \left\| \gamma(\varsigma) \right\| \le l_{\gamma}, &
         \left\| \kV(\cvar) \right\| \le l_{k}
      \end{matrix}.
   \end{equation*}
   From the previous arguments follows that $\left\| d \kV(\cvar) / d \cvar \right\| \le l_{dk}$, yielding to
   $\left\| \dot{\sigma}^2 \right\| \le l_{\sigma}$ with $l_{\sigma} = l_{\cvar}l_{dk}l_{\gamma}l_{k}$.
   Note that the bound $l_{\sigma}$ does not depend neither on the chosen regulator parameters, nor on the initial conditions.
   Consider now the jump dynamics, under the~\assmref{assm:K-CONTINUOUS-BOUNDED} of Lipschitz continuous kernel, we can explicitly derive an upper bound on
   the value of $\left(\sigma^2\right)^+$ at each jump (see~\cite[Theorem 1]{lederer2021uniform})
   \begin{equation*}
      \begin{split}
         \left(\sigma^2\right)^+ \le \Big[ \big|\B(\bar{\cvar}) & \big| \left(\kF(\cvar^*, \cvar^*) + 2L_{k} \varrho \right) + \sigma_n^2 \Big]^{-1} \\
         & \Big[ \left(2lb \big( \kF(\cvar^+, \cvar^+ ) + \kF(\cvar^*, \cvar^+)\right) \\ 
         & - L_{k}^2 \varrho^2\big) \big|\B(\cvar^*)\big| + \sigma_n^2 \kF(\cvar^+, \cvar^+)\Big],
      \end{split}
   \end{equation*}
   where $\B(\cvar^*)$ denotes the training data set restricted to a ball around $\cvar^*$ with radius $\varrho \in \R$.
   In our specific case $\cvar^* = \cvar^+$, thus the aforementioned relation boils down to
   \begin{equation*}
      \left(\sigma^2\right)^+ \le \frac{\left(4L_k \varrho \sigma_p^2 - L_k^2 \varrho^2\right)m_{\B(\cvar^*)} + \sigma_n^2 \sigma_p^2}{\left(\sigma_p^2 + L_k \varrho\right)m_{\B(\cvar^*)} + \sigma_n^2}
   \end{equation*}
   with $\varrho \le \sigma_p^2 / L_k$ and $m_{\B(\cvar^*)} \le m$ is the cadinality of $\B(\cvar^*)$.
   Clearly, the worst case scenario is represented by the case in wich $\B(\cvar^*)$ cannot embrace any training points, thus when $\varrho = 0$.
   In this particular case we get
   \begin{equation*}
      \left(\sigma^2\right)^+ \le \frac{\sigma_n^2 \sigma_p^2}{\sigma_p^2 + \sigma_n^2} < \ \sigma^2_{\text{thr}}.
   \end{equation*}
   Finally, the existence of $\underline{\text{T}}$ follows from the fact that the flow dynamics~\eqref{eq:SIGMA-FLOW} is continuous
   with upper bounded norm and its initial condition, at each jump time, is lower than the chosen threshold $\sigma^2_{\text{thr}}$.
   This ensures persistency of flow intervals.
\end{proof}
\begin{lemma}
   Consider the hybrid subsystem
   \begin{equation}%
      \label{eq:HYBRID-SUBSYSTEM}
      \begin{split}
         &
         \begin{cases}
            \dot{\tau} & = 1 \\
            \dot{\eta} & = \Phi(\eta, \mu(\eta, \varsigma, \alpha, \tau)) + Ge \\
            \dot{\varsigma} & = 0 \\
         \end{cases} \\
         & (\eta, \varsigma, \tau) \in \C, \\
         &
         \begin{cases}
            \tau^{+} & = 0 \\
            \eta^{+} & = \eta \\
            \varsigma^{+} & = (S \otimes I_{p}) \varsigma + (B \otimes I_{p}) \\ 
            & \hspace{1.5cm} \begin{bmatrix} \eta+\delta_{\eta_{1:d-1}} & \eta_d+\delta_{\eta_d} & \tau \end{bmatrix}\T \\
         \end{cases} \\
         & (\eta, \varsigma, \tau) \in \D, \\
      \end{split}
   \end{equation}
   with $\C$ and $\D$ defined as above and $\delta_{\eta} = (\delta_{\eta_{1:d-1}}, \delta_{\eta_d}) \in \R^{dn_e}$ hybrid input.
   Let Assumptions~\ref{assm:MU-CONTINUOUS-BOUNDED} and~\ref{assm:K-CONTINUOUS-BOUNDED}, and Lemma~\ref{lem:DWELL-TIME} hold, then the tuple $\left(\cS, \mu, \gamma, \alpha\right)$
   satisfies the identifier requirements~\cite{bin2020approximate} relative to $\cJ$, namely there exists a compact set $S^* \subset \cS$, $\beta_{\varsigma} \in \KL$,
   a Lipschitz function $\rho_{\varsigma} \in \cK$, and for each solution $(w, x, \eta, \tau)$ to Equations~\eqref{eq:GENERAL-EXOSYSTEM},~\eqref{eq:PARTICULAR-SYSTEM}, and~\eqref{eq:INTERNAL-MODEL},
   a hybrid arc $\varsigma^*: \text{dom}\left(w, x, \eta, \tau\right) \mapsto \cS$ and a $j^* \in \N$ such that $(\tau, \eta, \varsigma^*, d)$ with $\delta_{\eta} = 0$
   is a solution pair to~\eqref{eq:HYBRID-SUBSYSTEM} satisfying $\varsigma^*(j) \in S^*$ for all $j \ge j^*$, and the following holds:
   \begin{enumerate}
      \item \textit{\textbf{Optimality:}}
      For all $j \ge j^*$, the function $\mu^*(\cdot) = \mu(\eta, \varsigma^*, \gamma(\varsigma^*), \tau)$ satisfies
      \begin{equation*}
         \mu^*_j(\cdot) \in \arg \min \cJ_j.
      \end{equation*}
      \item \textit{\textbf{Stability:}}
      For every hybrid input $\delta_{\eta}$, every solution pair $(\eta, \varsigma, \tau, \delta)$ of the hybrid subsystem~\eqref{eq:HYBRID-SUBSYSTEM}
      satisfies for all $j \in \text{Jmp}(\eta, \varsigma, \tau)$
      \begin{equation*}
         \left| \varsigma(j) - \varsigma^*(j) \right| \le \max \{ \beta_{\varsigma}(\left| \varsigma(0) - \varsigma^*(0) \right|,j), \rho_{\varsigma}(\left| \delta_{\eta} \right|_j) \}.
      \end{equation*}
      \item \textit{\textbf{Regularity:}}
      The map $\mu(\cdot)$ is Lipschitz and differentiable with a locally Lipschitz derivative.
   \end{enumerate}
\end{lemma}
The proof of the latter lemma directly follows from the same arguments proposed by~\cite[Proposition 2]{bin2020approximate}.

\section{Numerical Simulations}%
\label{sec:NUMERICAL-SIMULATIONS}
\begin{figure}[t]
	\centering
	\includegraphics[width=0.35\textwidth]{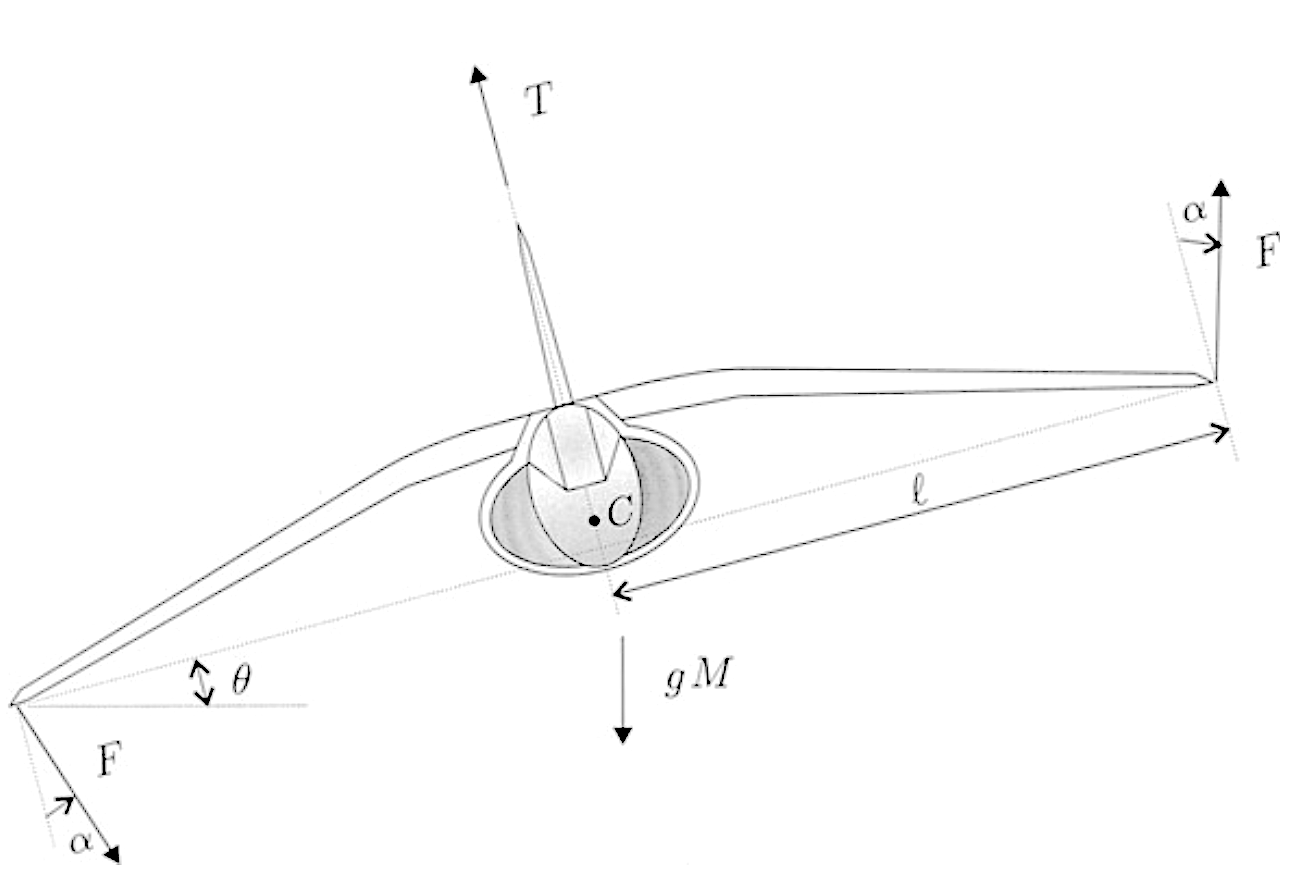}
	\caption{The vertical takeoff and landing aircraft considered in numerical simulations.}%
   \label{fig:VTOL-MODEL}
\end{figure}
\begin{figure}[t]
	\centering
	\includegraphics[trim={1cm 0.8cm 0cm 0cm}, width=0.5\textwidth]{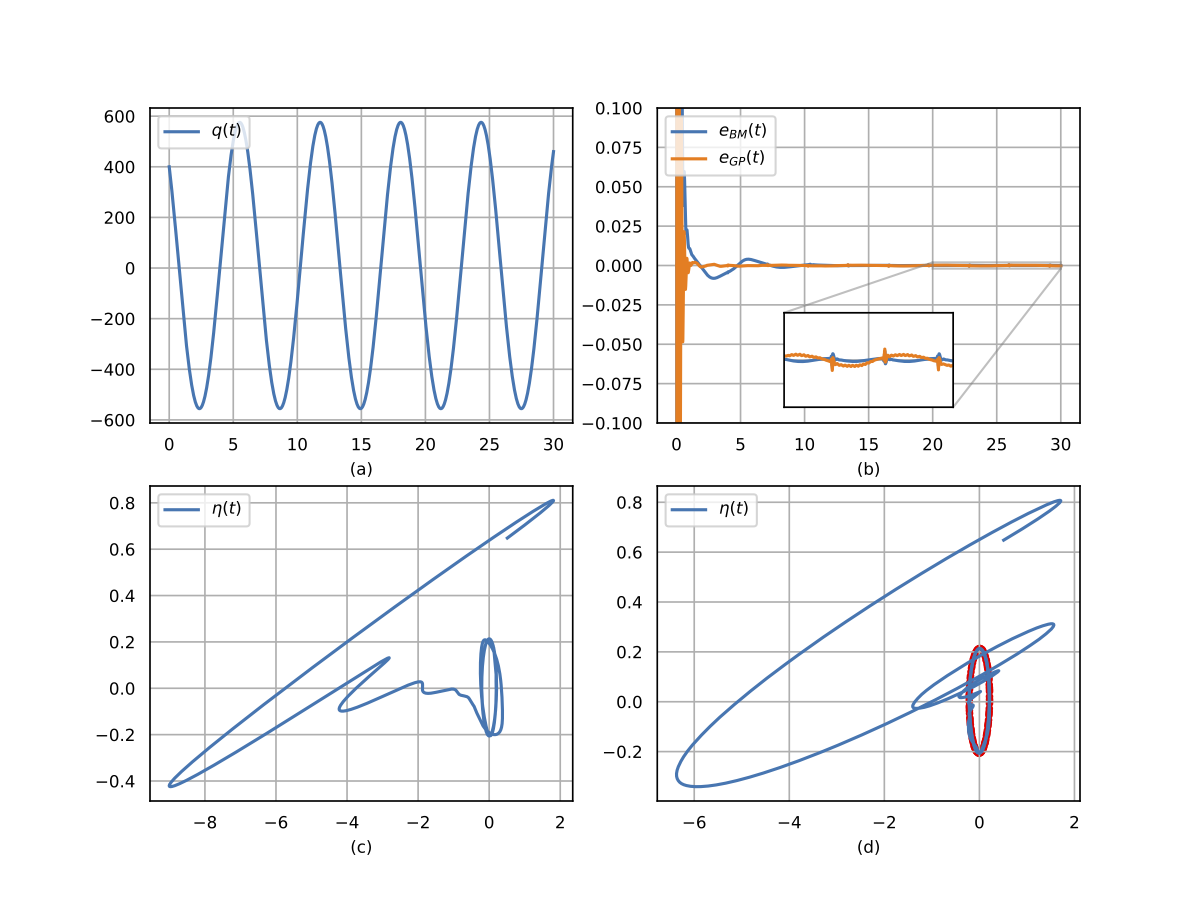}
	\caption{Results obtained comparing our approach $(e_{GP})$ versus~\cite{bin2019class} $(e_{BM})$ when the exogenous disturbance $d(w)$ is generated by~\eqref{eq:EXOSYSTEM-1}.
   In both cases, the used regulator parameters are the same reported in~\tabref{tab:REGULATOR-SIMULATION-PARAMETERS}.
   In the figure, image $(a)$ depicts the injected noise, $(b)$ compares the behavior of the regulation errors, while $(c)$ and $(d)$ shows the dynamics of $(\eta_1, \eta_2)$
   along the experiments in which the Bin-Marconi regulator~\cite{bin2019class} and ours is applied, respectively. In figure $(c)$ the used samples $(\varsigma)$ during the last
   flow interval are shown as red dots.}%
   \label{fig:LIN-SIM}
\end{figure}
We consider, as a testbed, the problem of regulation the lateral $\left(y_1, y_2\right)$ and angular $\left(\theta_1, \theta_2\right)$
dynamics of a Vertical-TakeOff-and-Landing (VTOL) aircraft~\cite{isidori2003robust} subjected to lateral forces
produced by the wind denoted by $d(w)$. A graphical representation of the considered system is reported in~\figref{fig:VTOL-MODEL}.
The VTOL dynamics reads as
\begin{equation}%
   \label{eq:VTOL-DYNAMICS}
   \begin{matrix}
      \begin{split}
         \dot{y}_1 & = y_2,  \\ \dot{y}_2 & = d(w) - \boldsymbol{g} \tan(\theta_1) + v,
      \end{split} & 
      \begin{split}
         \dot{\theta}_1 & = \theta_2, \\ \dot{\theta}_2 & = 2\boldsymbol{l}J^{-1}u,
      \end{split}
   \end{matrix}
\end{equation}
where $M>0$ and $J>0$ are the aircraft mass and inertia respectively, while $\boldsymbol{l}>0$ represents the wings lenght and $\boldsymbol{g} > 0$ the
gravitational constant. The input $u$ is the force $\left(F\right)$ on the wingtips, $v$ is a vanishing input taking into
account the (controlled) vertical dynamics (not considered here), and $d(w) := M^{-1}d_0(w)$, with $d_0(w)$ that is
the lateral force produced by the wind. Considering as regulation error the aircraft lateral position $\left(e = y_1\right)$, the control
objective is to remove the wind disturbance out from the lateral dynamics.
Let $w(t)$ be generated by an exosystem of the form~\eqref{eq:GENERAL-EXOSYSTEM} and consider the following change of coordinates
\begin{equation*}
   \begin{matrix}
      \begin{split}
         \chi_1 & = y_1, \\ \chi_2 & = y_2,
      \end{split} &
      \begin{split}
         \chi_3 & = d(w) + \boldsymbol{g} \tan(\theta_1), \\ \zeta & = L_s d(w) - {\boldsymbol{g} \theta_2}/{\cos(\theta_1)}.
      \end{split}
   \end{matrix}
\end{equation*}
In the new coordinates, letting $x = \col(\chi, \zeta)$, the system~\eqref{eq:VTOL-DYNAMICS} reads as follow
\begin{equation}%
   \label{eq:NEW-VTOL-SYSTEM}
   \begin{matrix}
      \begin{split}
         \dot{\chi}_1 & = \chi_2, \\
         \dot{\chi}_2 & = \chi_3,
      \end{split} &
      \begin{split}
         \dot{\chi}_3 & = \zeta, \\
      \dot{\zeta} & = q(w,x) + \Omega(w,x)u,
      \end{split}
   \end{matrix}
\end{equation}
with $q(w,x)$ and $\Omega(w,x)$ given by
\begin{equation*}
   \begin{split}
      q(w,x) = L_s^2 d(w) - & \frac{1}{\boldsymbol{g}} \left(L_s d(w) - \zeta\right)^2 \\
      & \sin\left(2 \tan^{-1}\left(\frac{d(w) - \chi_3}{\boldsymbol{g}}\right)\right),
   \end{split}
\end{equation*}
\begin{equation*}
   \Omega(w,x) = -2\boldsymbol{g}\boldsymbol{l}J^{-1} \cos\left(\tan^{-1}\left(\frac{d(w) - \chi_3}{\boldsymbol{g}}\right)\right)^{-2}.
\end{equation*}
The system~\eqref{eq:NEW-VTOL-SYSTEM} is in the form~\eqref{eq:PARTICULAR-SYSTEM} with~\assmref{assm:REGULATOR-EQUATIONS} trivially fulfilled, since
the $x_0$ dynamics is absent, by $x^* = 0$ and $u^* = \left(\boldsymbol{g}L_s^2d(w) - 2d(w)(\boldsymbol{g}^2 + d(w)^2)(L_sd(w))^2\right)/2\boldsymbol{l}J^{-1}(\boldsymbol{g}^2 + d(w)^2)$, and~\assmref{assm:UNIFORM-DETECTABILITY}
fulfilled on each compact set with $\mathcal{L}$ a negative number.
With $\left(c_1,c_2,c_3\right)$ the coefficients of a Hurwitz polynomial and $\delta,l > 0$ design parameters, we fix the control law as
\begin{equation*}
   \begin{split}
      u = \mathcal{L}\Big[ c_1 & l \delta^3 (y_1 + \eta_1) + c_2  l \delta^2 y_2 \\ 
         & + c_3 l \delta (-\boldsymbol{g} \tan(\theta_1)) + l(-\boldsymbol{g} \theta_2 / \cos^2(\theta_1)) \Big].
   \end{split}
\end{equation*}
Figures~\ref{fig:LIN-SIM},~\ref{fig:POW-SIM}, and~\ref{fig:ATAN-SIM} report the obtained results when the aircraft is perturbed with
a lateral disturbance $d(w) = (2(10^7w_1) + 10^6w_3)/M$, where $w_1$ and $w_3$ are the states of three different exosystems.
In particular, in~\figref{fig:LIN-SIM}, $s(w)$ reads as
\begin{figure}[t]
	\centering
	\includegraphics[trim={1cm 0.8cm 0cm 0cm}, width=0.5\textwidth]{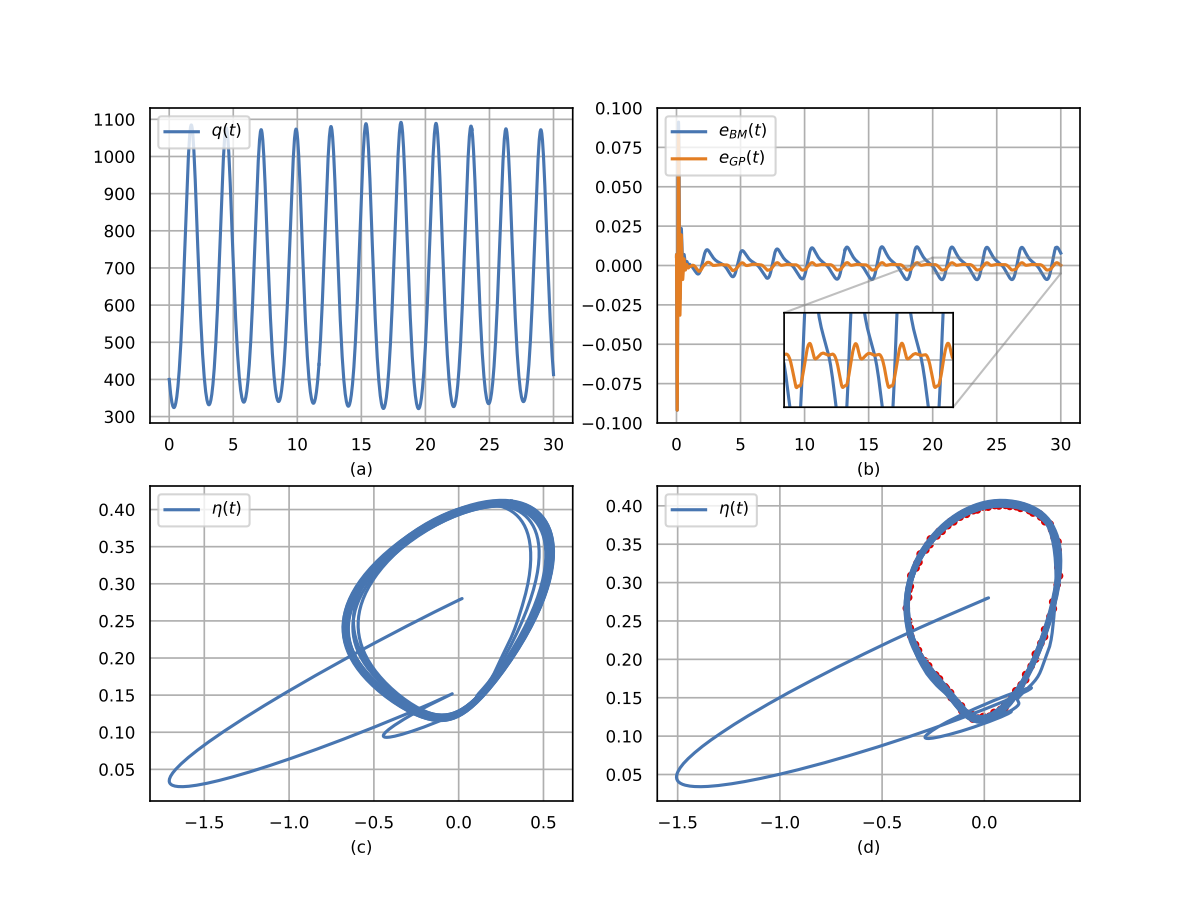}
	\caption{Results obtained comaring our approach $(e_{GP})$ versus~\cite{bin2019class} $(e_{BM})$ when the exogenous disturbance $d(w)$ is generated by~\eqref{eq:EXOSYSTEM-2}.
   For a comprehensive explanation of Figures $(a)$, $(b)$, $(c)$, and $(d)$ please refer to~\figref{fig:LIN-SIM}}%
   \label{fig:POW-SIM}
\end{figure}
\begin{equation}%
   \label{eq:EXOSYSTEM-1}
   \begin{matrix}
      \begin{split}
         \dot{w}_1 & = w_2, \\
         \dot{w}_2 & = -w_1, \\
      \end{split} &
      \begin{split}
         \dot{w}_3 & = w_4, \\
         \dot{w}_4 & = -4w_3,
      \end{split}
   \end{matrix}
\end{equation}
in~\figref{fig:POW-SIM}, $s(w)$ behaves as
\begin{equation}%
   \label{eq:EXOSYSTEM-2}
   \begin{matrix}
      \begin{split}
         \dot{w}_1 & = w_2, \\
         \dot{w}_2 & = 4w_1 - w_1^3, \\
      \end{split} &
      \begin{split}
         \dot{w}_3 & = w_4, \\
         \dot{w}_4 & = -4w_3,
      \end{split}
   \end{matrix}
\end{equation}
while in~\figref{fig:ATAN-SIM}, the exosystem is described by
\begin{equation}%
   \label{eq:EXOSYSTEM-3}
   \begin{matrix}
      \begin{split}
         \dot{w}_1 & = w_2, \\
         \dot{w}_2 & = 3\tan^{-1}\left(w_1\right) - w_1, \\
      \end{split} &
      \begin{split}
         \dot{w}_3 & = w_4, \\
         \dot{w}_4 & = -4w_3.
      \end{split}
   \end{matrix}
\end{equation}
In all simulations we exploit the same set of parameters, for both the regulator and the discrete-time identifier.
The adopted parameters are reported in~\tabref{tab:GP-SIMULATION-PARAMETERS},~\tabref{tab:REGULATOR-SIMULATION-PARAMETERS}, and~\tabref{tab:MODEL-SIMULATION-PARAMETERS}.
\begin{figure}[t]
	\centering
	\includegraphics[trim={1cm 0.8cm 0cm 0cm}, width=0.5\textwidth]{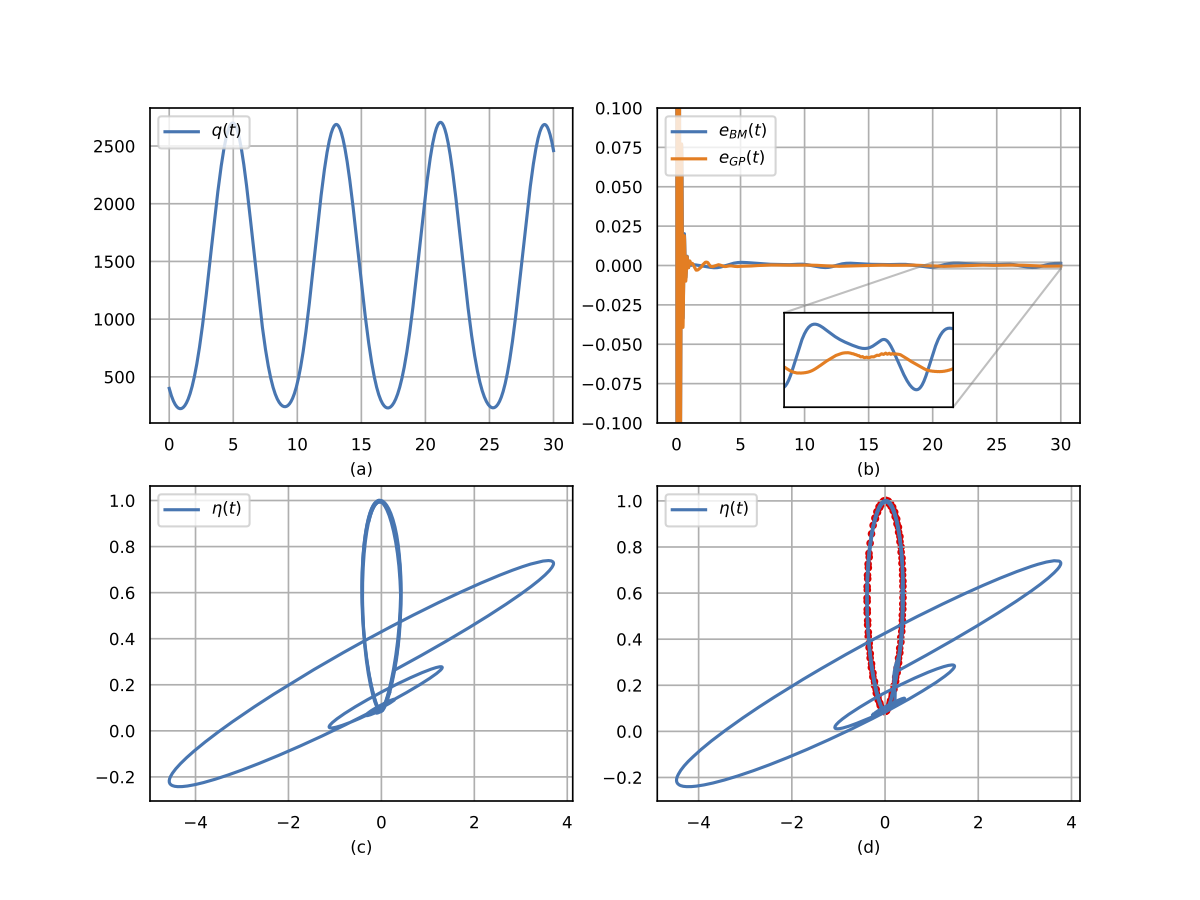}
	\caption{Results obtained comaring our approach $(e_{GP})$ versus~\cite{bin2019class} $(e_{BM})$ when the exogenous disturbance $d(w)$ is generated by~\eqref{eq:EXOSYSTEM-3}.
   For a comprehensive explanation of Figures $(a)$, $(b)$, $(c)$, and $(d)$ please refer to~\figref{fig:LIN-SIM}}%
   \label{fig:ATAN-SIM}
\end{figure}
\begin{table}[b!]
   \centering
   \begin{tabular}{|c|c|c|c|c|c|c|}
       \hline
       \hline
       $n_{\text{ds}}$ & $\lambda_{\eta_1}$ &  $\lambda_{\eta_2}$ & $\lambda_{\tau}$ & $\sigma^2_p$ & $\sigma^2_n$ & $\sigma^2_{\text{thr}}$ \\
       \hline
       $100$ & $0.1$ & $0.1$ & $2$ & $1$ & $0.01$ & $0.1$\\
       \hline
       \hline
   \end{tabular}
   \caption{Gaussian process parameters used in simulations.}%
   \label{tab:GP-SIMULATION-PARAMETERS}
\end{table}
\begin{table}[b!]
   \centering
   \begin{tabular}{|c|c|c|c|c|c|c|c|}
       \hline
       \hline
       $\left(c_1, c_2, c_3\right)$ & $l$ & $\delta$ & $\mathcal{L}$  & $\left(h_1, h_2\right)$ & $g$ & $\left(m_1, m_2\right)$ & $\rho$ \\
       \hline
       $\left(15, 75, 125\right)$ & $250$ & $150$ & $20$ & $\left(15, 70\right)$ & $2$ & $\left(20, 20\right)$ & $2$ \\
       \hline
       \hline
   \end{tabular}
   \caption{Regulator parameters used in simulations.}%
   \label{tab:REGULATOR-SIMULATION-PARAMETERS}
\end{table}
\begin{table}[b!]
   \centering
   \begin{tabular}{|c|c|c|c|}
       \hline
       \hline
       $M$ & $J$ & $\boldsymbol{l}$ & $\boldsymbol{g}$ \\
       \hline
       $5\cdot 10^4$ & $1.25 \cdot 10^4$ & $2$ & $9.81$ \\
       \hline
       \hline
   \end{tabular}
   \caption{Model parameters used in simulations.}%
   \label{tab:MODEL-SIMULATION-PARAMETERS}
\end{table}

\section{Conclusions}%
\label{sec:CONCLUSIONS}
We presented an adaptive learninm-based technique to design internal model-based regulator for a large class
of nonlinear systems. The technique fits in the general framework recently proposed in~\cite{bin2019class} and
shows how the identification of the optimal steady state control input can be performed by using Gaussian process
models. As opposite to previous approaches no immersion assumptions into specific model sets are 
assumed and only smoothness of the ideal steady state control input is required. The learning-based adaptation
is performed by following an ``event-triggered'' logic and hybrid tools are used in the analysis of the closed-loop system.
The paper also presents numerical simulations showing how the proposed method outperforms previous approaches when
the regulated plant or the exogenous disturbances are subject to unmodeled perturbations.

\bibliographystyle{unsrt}
\bibliography{root.bib}
\end{document}